\documentclass[a4paper,autoref]{lipics-v2021}
\usepackage{amsmath,amssymb,amsfonts,latexsym}
\usepackage{color,enumerate,graphicx}
\usepackage{url,hyperref}
\usepackage{fancybox}
\usepackage{algorithm}
\usepackage[noend]{algpseudocode}
\usepackage{mathtools}
\usepackage{xspace} 
  \nolinenumbers
\hideLIPIcs
\definecolor{light-gray}{gray}{0.95}
\floatstyle{plaintop}
\newcommand{\eps}{\varepsilon}
\newcommand{\etal}{\emph{et al.}\xspace}

\theoremstyle{plain}

\newcommand{\B}{\ensuremath{\mathcal{B}}}
\newcommand{\C}{\ensuremath{\mathcal{C}}}

\newcommand{\R}{\ensuremath{\mathcal{R}}}

\newcommand{\REAL}{\ensuremath{\mathbb{R}}}
\newcommand{\Reals}{\REAL}

\renewcommand{\leq}{\leqslant}
\renewcommand{\geq}{\geqslant}

\DeclarePairedDelimiter\floor{\lfloor}{\rfloor}

\DeclareMathOperator{\diam}{diam}
\DeclareMathOperator{\radius}{radius}
\DeclareMathOperator{\dist}{dist}

\newcommand{\diamz}{\diam_z}
\DeclareMathOperator{\cent}{center}
\DeclareMathOperator{\ball}{ball}

\newcommand{\opt}{\mbox{{\sc opt}}\xspace}
\newcommand{\myopt}{\mathrm{opt}}
\newcommand{\optkz}{\opt_{k,z}}

\newcommand{\NO}{{\sc no}\xspace}

\newcommand{\ts}{\tau}
\newcommand{\et}{t_{\mathrm{exp}}}
\newcommand{\Texp}{T_{\mathrm{exp}}}
\newcommand{\at}{t_{\mathrm{arr}}}
\newcommand{\sketch}{\Gamma}

\newcommand{\answer}{\mathit{answer}}
\newcommand{\myin}{\mathrm{in}}
\newcommand{\myout}{\mathrm{out}}
\newcommand{\Pin}{P_{\myin}}
\newcommand{\Pout}{P_{\myout}}
\newcommand{\alg}{{\sc Alg}\xspace}
\newcommand{\dmax}{\Delta_{\max}}
\newcommand{\dmin}{\Delta_{\min}}
\newcommand{\ropt}{\rho_{\myopt}}

\title{$k$-Center Clustering with Outliers in the Sliding-Window Model}

\funding{MdB is supported by the  Dutch Research Council (NWO) through    Gravitation-grant NETWORKS-024.002.003.}

\author{Mark de Berg}{Department of Mathematics and Computer Science, TU Eindhoven, the Netherlands}{M.T.d.Berg@tue.nl}{}{}
\author{Morteza Monemizadeh}{Department of Mathematics and Computer Science, TU Eindhoven, the Netherlands}{M.Monemizadeh@tue.nl}{}{}
\author{Yu Zhong}{Department of Mathematics and Computer Science, TU Eindhoven, the Netherlands}{Y.Zhong@student.tue.nl}{}{}

\authorrunning{M.~de Berg, M.~Monemizadeh and Y.~Zhong} %mandatory. First: Use abbreviated first/middle names. Second (only in severe cases): Use first author plus 'et. al.'

\Copyright{Mark de Berg and Morteza Monemizadeh and Yu Zhong}%mandatory, please use full first names. LIPIcs license is "CC-BY";  http://creativecommons.org/licenses/by/3.0/


% mandatory: Please provide 1-5 keywords
\EventEditors{Petra Mutzel, Rasmus Pagh, and Grzegorz Herman}
\EventNoEds{3}
\EventLongTitle{29th Annual European Symposium on Algorithms (ESA 2021)}
\EventShortTitle{ESA 2021}
\EventAcronym{ESA}
\EventYear{2021}
\EventDate{September 6--8, 2021}
\EventLocation{Lisbon, Portugal}
\EventLogo{}
\SeriesVolume{204}
\ArticleNo{13}
\begin{document}
\maketitle
%------------------------------------------------------------------------------------------

%------------------------------------------------------------------------------------------
\begin{abstract}
The $k$-center problem for a point set~$P$ asks for a collection of $k$ congruent balls
(that is, balls of equal radius) that together cover all the points in~$P$ and whose 
radius is minimized. The $k$-center problem with outliers is defined similarly, except 
that $z$ of the points in $P$ do need not to be covered, for a given parameter~$z$. We study the
$k$-center problem with outliers in data streams in the sliding-window model.
In this model we are given a possibly infinite stream 
$P=\langle p_1,p_2,p_3,\ldots\rangle$ of points and a time window of length~$W$,
and we want to maintain a small sketch of the set $P(t)$ of points currently 
in the window such that using the sketch we can approximately solve the 
problem on~$P(t)$.

We present the first algorithm for the $k$-center problem with outliers in the 
sliding-window model. The algorithm  works for the case where the points come from a
space of bounded doubling dimension and it maintains a set $S(t)$ such that
an optimal solution on $S(t)$ gives a $(1+\eps)$-approximate solution on $P(t)$.
The algorithm is deterministic and uses $O((kz/\eps^d)\log \sigma)$ storage, where $d$ is the doubling
dimension of the underlying space and $\sigma$ is the spread of the points in the stream.
Algorithms providing a $(1+\eps)$-approximation were not even known in the setting
without outliers or in the insertion-only setting with outliers.
We also present a lower bound showing that any algorithm that provides a $(1+\eps)$-approximation
must use $\Omega((kz/\eps)\log \sigma)$ storage.
\end{abstract}
%------------------------------------------------------------------------------------------

%------------------------------------------------------------------------------------------
\section{Introduction}
%------------------------------------------------------------------------------------------
Clustering is one of the most important tools to analyze large data sets.
A well-known class of clustering algorithms is formed by centroid-based algorithms, which include
$k$-means clustering, $k$-median clustering and $k$-center clustering. The latter type
of clustering is the topic of our paper. In the $k$-center problem one is given a set $P$ of points 
from a metric space and a parameter~$k$, and the goal is to find $k$ congruent 
balls (that is, balls of equal radius) that together 
cover the points from~$P$ and whose radius is minimized. Note that the special 
case $k=1$ corresponds to the minimum-enclosing ball problem. 
% The $k$-center problem, including the special cases for $k=1$ and $k=2$, has been studied extensively, in particular  in the Euclidean setting. 
Data sets in practice often contain outliers, leading to the
\emph{$k$-center problem with outliers}. Here we are given, besides $P$ and $k$,
a parameter $z$ that indicates the allowed number of outliers. 
Thus the radius of the balls in an optimal solution is given by
\begin{quotation}
$\optkz(P)$ := \begin{minipage}[t]{100mm}
             the smallest radius $\rho$ such that we can cover all points from $P$,
             except for at most $z$ outliers, by $k$ balls of radius~$\rho$.
             \end{minipage}
\end{quotation}

In this paper we study the $k$-center problem with outliers in data streams, where the input is a possibly
infinite stream $P=\langle p_1,p_2,\ldots\rangle$ of points. The goal is to maintain a solution
to the $k$-center problem as the points arrive over time, without any knowledge of future
arrivals and using limited (sub-linear) storage. Since we cannot store all the points in the stream, 
we cannot expect to maintain an optimal solution. Hence, the two main quality criteria of
a streaming algorithm are its approximation ratio and the amount of storage it uses.
We will study this problem in the \emph{sliding-window model}. In this model we are given a window length~$W$ and we are, at any time $t$, only interested in the points that arrived in the time window~$(t-W,t]$.
Working in the sliding-window model is often significantly more difficult than working in
the standard (insertion-only) streaming model.

%------------------------------------------------------------------------------------------
\paragraph*{Previous work}
%------------------------------------------------------------------------------------------
Charikar~\etal\cite{DBLP:journals/siamcomp/CharikarCFM04} 
were the first to study the metric $k$-center problem in data streams. 
They developed an algorithm in the insertion-only model that computes 
an $8$-approximation for the $k$-center problem using $\Theta(k)$ space.  
Later McCutchen and Khuller~\cite{DBLP:conf/approx/McCutchenK08}
improved the approximation ratio to~$2+\eps$ at the cost of increasing
the storage to $O((k/\eps)\log(1/\eps))$.
McCutchen and Khuller also studied the $k$-center problem with $z\geq 1$ outliers, 
for which they gave a $(4 + \eps)$-approximation algorithm 
that requires $O(kz/\eps)$ space.

The above results are for general metric spaces.  
In spaces of \emph{bounded doubling dimension}\footnote{The \emph{doubling dimension}
of a space $X$ is the smallest number~$d$ such that any ball~$B$ in the space can be covered
by $2^d$ balls of radius $\radius(B)/2$.} better bounds are possible. Indeed,
Ceccarello~\etal~\cite{DBLP:journals/pvldb/CeccarelloPP19} 
gave a $(3+\eps)$-approximation algorithm for the $k$-center problem with $z$ outliers,
thus improving the approximation ratio~$(4+\eps)$ for general metrics. Their algorithm
requires $O((k+z)(1/\eps)^d)$ storage, where $d$ is the doubling dimension of 
the underlying space (which is assumed to be a fixed constant).

%  using the idea of \emph{composable coresets} they developed a deterministic $1$-pass streaming algorithm 
% that requires $O(k+z)(1/\eps)^d$ space and approximates the $k$-center problem with $z$ outliers 
% to within $(3 + \eps)$-factor. 
%In the MapReduce model they developed deterministic $2$-round, $(2 + \eps)$-approximation algorithm 
%for the $k$-center problem that requires  $O(\sqrt{|S|k}(1/\eps)^d)$ space per machine. 
%For the $k$-center problem with $z$ outliers they devised 
%a deterministic $2$-round, $(3 + \eps)$-approximation MapReduce algorithm 
%with $O(\sqrt{|S|(k+\log|S|)+z}(1/\eps)^d)$ space per machine. 
%For the $k$-center problem with $z$ outliers in the streaming model, 
%they showed a deterministic $1$-pass, $(3 + \eps)$-approximation algorithm 
%which requires $O(k+z)(1/\eps)^d$ space.

%\textbf{$k$-center with relaxed outliers.} 
The algorithms mentioned so far are deterministic. 
Charikar~\etal~\cite{DBLP:conf/stoc/CharikarOP03} 
and Ding~\etal~\cite{DBLP:conf/esa/DingYW19} studied sampling-based
streaming algorithms for the Euclidean $k$-center problem with outliers,
showing that if one allows slightly more than $z$~outliers then
randomization can reduce the storage requirements. Our focus, however,
is on deterministic algorithms.
%
%for the $k$-center clustering with outliers in Euclidean space $\Reals^d$ of much smaller space. 
%Indeed, suppose that we are given a point set $P \subseteq \Reals^d$, 
%and we sample a subset $S \subseteq P$  of size $\tilde{O}(\frac{1}{\eps^2\gamma}\cdot kd)$ uniformly
%at random and solve $k$-center clustering with $z'=(1+\eps)\gamma|S|$ outliers on $S$. 
%Then, if $X$ is an $\alpha$-approximate solution of the $k$-center problem with $(1+\eps)z'$ 
%for $S$, then $X$ is an $\alpha$-approximate solution of the $k$-center problem 
%with $(1+c\eps)z$ for $P$, with a constant probability for a constant $c$.
\medskip

For the $k$-center problem in the sliding-window model, the only result we are 
aware of is due to Cohen{-}Addad~\etal~\cite{DBLP:conf/icalp/Cohen-AddadSS16}. 
They deal with the $k$-center problem in general metric spaces, but without outliers,  
and they propose a $(6 + \eps)$-approximation algorithm using 
$O((k/\eps) \log \sigma)$ storage, 
and a $(4 + \eps)$-approximation for the special case $k = 2$. 
Here $\sigma$ denotes (an upper bound on) the spread of the points in the stream;
in other words, $\sigma := \dmax/\dmin$, where $\dmax$ is an upper bound on the maximum distance
and $\dmin$ is a lower bound on the minimum distance between any two distinct points.
It is assumed that $\dmax$ and $\dmin$ are known to the algorithm.
They also prove that any algorithm for the $2$-center problem with outliers in general metric
spaces that achieves an approximation ratio of less than~$4$ requires $\Omega(W^{1/3})$ space, 
where $W$ is the size\footnote{Here the window size $W$ is defined in terms of 
the number of points in the window, that is, the window consists of the $W$ 
most recent points. We define the window in a slightly more general manner,
by defining $W$ to be the length (that is, duration) of the window. Note that 
if we assume that the $i$-th point arrives at time $t=i$, then the two
models are the same.} of the window.
Table~\ref{table:streaming_MM_size} gives an overview of the known results
on the $k$-center problem in the insertion-only and the sliding-window model.
%------------------------------------------------------------------------------------------
\begin{table*}[t]
\begin{center}
\begin{tabular}{c|ccccc}
%\hline
% \multicolumn{4}{c}{Model} \\
%\cline{2-3}
model & metric space  &  approx. & storage  & outliers & ref. \\
\hline
insertion-only & general &   $8$ & $k$  &  no & \cite{DBLP:journals/siamcomp/CharikarCFM04} \\
        & general &    $2+\eps$ & $(k/\eps)\log(1/\eps)$  &  no & \cite{DBLP:conf/approx/McCutchenK08} \\     
        & general &  $4+\eps$ & $kz/\eps$  &  yes & \cite{DBLP:conf/approx/McCutchenK08}  \\     
\cline{2-6}
        & bounded doubling &   $3+\eps$ & $(k+z)/\eps^d$  &  yes & \cite{DBLP:journals/pvldb/CeccarelloPP19} \\     
\hline
sliding window & general  & $6+\eps$ & $(k/\eps)\log \sigma$  &  no  & \cite{DBLP:conf/icalp/Cohen-AddadSS16}  \\     
\cline{2-6}
& bounded doubling                            &     $1+\eps$ & $(kz/\eps^d)\log \sigma$ & yes    & here       \\
%\hline
\end{tabular}
\end{center}
\caption{Results for the $k$-center problem with and without outliers in the insertion-only
and the sliding-window model. Bounds on the storage are asymptotic. In the papers where the metric space has 
bounded doubling dimension or is Euclidean, the dimension~$d$ is considered a constant.}
\label{table:streaming_MM_size}
\end{table*}
\medskip

While our main interest is in the $k$-center problem for $k>1$, our results also
apply when~$k=1$. Hence, we also
briefly discuss previous results for the 1-center problem.

For the 1-center problem in $d$-dimensional Euclidean space,
streaming algorithms that maintain an $\eps$-kernel  give a $(1+\eps)$-approximation. 
An example is the algorithm of Zarabi-Zadeh~\cite{z-cpa-08} which maintains 
an $\eps$-kernel of size $O(1/\eps^{(d-1)/2}\log(1/\eps))$.
Moreover, using only $O(d)$ storage one can obtain a 1.22-approximation
for the 1-center problem without outliers~\cite{DBLP:journals/algorithmica/AgarwalS15,cp-sdameb-14}.
For the 1-center problem with outliers, one can obtain 
a $(1+\eps)$-approximation algorithm that uses $z/\eps^{O(d)}$ storage
by the technique of Agarwal~\etal~\cite{DBLP:journals/dcg/AgarwalHY08}. 
% and Har{-}Peled and Wang~\cite{DBLP:journals/siamcomp/Har-PeledW04} 
% introduced the idea of \emph{robust kernel} using which they 
%   \mdb{Both papers have the same result? Rephrase as: "using the techniques from ... one can obtain ...?}
Zarrabi{-}Zadeh and Mukhopadhyay ~\cite{DBLP:conf/cccg/Zarrabi-ZadehM09} studied the 
$1$-center problem with $z$~outliers in high-dimensional Euclidean spaces,
where $d$ is not considered constant, giving a $1.73$-approximation algorithm that requires $O(d^3z)$ storage. 
Recently, Hatami anad Zarrabi{-}Zadeh~\cite{DBLP:journals/comgeo/HatamiZ17} 
extended this result to $2$-center problem with $z$ outliers, obtaining
a $(1.8 + \eps)$-approximation using $O(d^3z^2+dz^4/\eps)$ storage. 
%This improves over the  $(4+\eps)$-approximation algorithm due to 
%McCutchen and Khuller~\cite{DBLP:conf/approx/McCutchenK08} for general
%metric spaces (and which also works for $k>2$).
None of the 1-center algorithms discussed above works in the sliding-window model.

A problem that is closely related to the 1-center problem is the
\emph{diameter problem}, where the goal is to maintain an approximation
of the diameter of the points in the stream. This problem has been
studied in the sliding-window model by Feigenbaum, Kannan, Zhang~\cite{DBLP:journals/algorithmica/FeigenbaumKZ04} and later
by Chan and Sadjad~\cite{cs-gosw-06}, who gave a $(1+\eps)$-approximation
for the diameter problem (without outliers) in the sliding window model, 
using $O((1/\eps)^{(d+1)/2}\log(\sigma/\eps))$ storage.

%------------------------------------------------------------------------------------------
\paragraph*{Our contribution}
%------------------------------------------------------------------------------------------
We present the first algorithm for the $k$-center problem with $z$ outliers in the sliding-window model.
It works in spaces of bounded doubling dimension and yields a $(1+\eps)$-approximation.
So far a $(1+\eps)$-approximation was not even known for the $k$-center problem
without outliers in the insertion-only model.  Our algorithm uses
$O((kz/\eps^d)\log \sigma)$ storage,\footnote{To correctly state the bound for the case $z=0$ as well,
we should actually write $O((k(z+1)/\eps^d)\log \sigma)$. To keep the bound concise, we will
just write $O((kz/\eps^d)\log \sigma)$, however.}
% \footnote{Since the solution also works in the case without outliers, it is more accurate to write $O((k(z+1)/\eps^d)\log \sigma)$, but for simplicity we write $O((kz/\eps^d)\log\sigma)$}
where $d$ is the doubling dimension and $\sigma$ is the spread,
as defined above. Thus for the 1-center problem we obtain a solution that uses 
$O((z/\eps^d)\log \sigma)$ storage. This solution also works for the diameter
problem. Note that also for the 1-center problem 
with outliers (and the diameter problem with outliers) an algorithm
for the sliding-window model was not yet known.
A useful property of the 
sketch\footnote{We use the word \emph{sketch} even though we do not study how to compose the sketches 
for two separate streams into a sketch for the concatenation of the stream, as
the term \emph{sketch} seems more appropriate than \emph{data structure}, for example.} 
maintained by our algorithm for the $k$-center problem, is that it can
also be used for the $k'$-center problem for any $k'< k$, as well as for
the diameter problem. 
As in the previous papers on the $k$-center problem (or the diameter problem) in the sliding-window 
model~\cite{cs-gosw-06,DBLP:conf/icalp/Cohen-AddadSS16,DBLP:journals/algorithmica/FeigenbaumKZ04}, we assume
that $\dmax$ and $\dmin$ (upper and lower bounds on the maximum and minimum distance between
any two points in the stream) are known to the algorithm. 
% A typical example is when the input consists of points in $d$-dimensional Euclidean space with 
% integer coordinates in a given range~$\{0,1,\ldots,U\}$. In this example, the spread is~$\Theta(U)$.

As mentioned above, our algorithm provides a $(1+\eps)$-approximation. More precisely,
it maintains a set $S(t)\subset P(t)$ such that computing an optimal solution for~$S(t)$ and 
then suitably expanding the balls in the solution, gives a $(1+\eps)$-approximate solution for~$P(t)$.
Thus, to obtain a $(1+\eps)$-approximate solution for $P(t)$ one needs to compute an optimal 
(or $(1+\eps')$-approximate, for a suitable~$\eps'$) solution for~$S(t)$. Depending on the 
underlying space this can be slow, and it may be preferable to compute, say, a 2-approximation for $S(t)$,
which would then give a $(2+\eps)$-approximation for~$P(t)$. The various options in Euclidean spaces
and general spaces of bounded doubling dimension are discussed in Section~\ref{subse:report-solution}.

Our second contribution is a lower bound for the $k$-center problem with outliers in the sliding-window model. 
% Recall that Ceccarello~\etal~\cite{ceccarello2018solving} provide a
% $(3+\eps)$-approximation using $O((k+z)/\eps^d)$ storage in the insertion-only model.
% MdB: Should perhaps compare with the Ceccarello paper, but then we need to check in which model it works.
Our lower bound shows that any algorithm that provides a $(1+\eps)$-approximation must
use $\Omega((kz/\eps)\log\sigma)$ storage. This matches our upper bound up to the dependency on~$\eps$ and~$d$. 
The lower-bound construction uses points in~$\Reals^1$, so it shows that our algorithm is optimal in this case.
The lower bound model is very general. It allows the algorithm to store points,
or weighted points, or balls, or whatever it wants so that it can approximate an optimal
solution for $P(t)$ at any time~$t$. The main condition is that each object comes with an expiration time
which corresponds to the expiration time of some point in the stream, and that the algorithm
can only change its state when an object expires or a new object arrives; see Section~\ref{se:lower-bound}
for details.

%------------------------------------------------------------------------------------------
\section{The algorithm}
%------------------------------------------------------------------------------------------
Let $P:= \langle p_1,p_2,\ldots \rangle$ be a possibly infinite stream of points from
a metric space~$X$ of doubling dimension~$d$ and spread~$\sigma$, where~$d$ is considered to
be a fixed constant. We denote the arrival time
of a point~$p_i$ by $\at(p_i)$. We say that $p_i$ \emph{expires} at 
time~$\et(p_i) := \at(p_i) + W$, where $W$ is the given length of the time window.
To simplify the exposition, we assume that all arrival times and departure times
(that is, times at which a point expires) are distinct.
For a time~$t$ we define $P(t)$ to be the set\footnote{We allow the same point from $X$ 
to occur multiple times in the stream, so $P(t)$ is actually a multi-set. Whenever we refer 
to ``sets'' in the remainder of the paper we mean ``multi-sets''. The distance $\dmin$
is defined with respect to distinct points, however.}
of points currently in the window. In other words, $P(t) := \{ p_i : \at(p_i) \leq t < \et(p_i) \}$.
For a point $q\in X$ and a parameter~$r\geq 0$, we use  $\ball(q,r)$ to denote the ball with center~$q$ 
and radius~$r$.

In the following we show how to maintain a sketch~$\sketch(t)$ of $P(t)$ for the $k$-center problem 
with outliers. The idea behind our algorithm is as follows. Consider an optimal solution for~$P(t)$
consisting of $k$ balls $B_1,\ldots,B_k$ of radius~$\ropt$. Suppose we cover the points in each ball~$B_i$ 
by a number of smaller balls of radius~$\eps\ropt$; we will call these \emph{mini-balls}.
% See Figure~\ref{fi:mini-balls}. 
%------------------------------------------------------------------------------------------
%\begin{figure}
%\begin{center}
%\includegraphics{Figures/mini-balls.pdf}
%\end{center}
%\caption{A set of $k=3$ balls of radius $\rho$ whose points are covered by mini-balls of radius~$\eps\rho$.
%\mdb{Figure does not seem very useful.}}
%\label{fi:mini-balls}
%\end{figure}
%------------------------------------------------------------------------------------------
Since the underlying space has doubling dimension~$d$,
we can do this using $O(1/\eps^d)$ mini-balls for each ball~$B_i$. To obtain a $(1+\eps)$-approximation, 
we do not need to know all the points: it suffices to keep $z+1$ points in each mini-ball,
since then we cannot designate all points in the mini-ball as outliers. There are several challenges
to overcome to make this idea work. For instance, we do not know~$\ropt$. Hence,
we will develop an algorithm for the decision version of the problem, which we will then 
apply ``in parallel'' to different possible values for the optimal radius.
But when $\rho$ is much smaller than $\ropt$ we have another challenge, namely that there
can be many outliers (which may need to be stored, because at some later moment
in time they may become inliers). Finally, since points arrive and expire, the balls in
an optimal solution move over time, which makes it hard to maintain the mini-balls 
correctly. Below we show how to overcome these challenges.

%------------------------------------------------------------------------------------------
\subsection{A sketch for the decision problem}
%------------------------------------------------------------------------------------------
Recall that for a set $Q$ of points, $\optkz(Q)$ denotes the smallest radius $\rho$ such that 
we can cover all points from $Q$, except for at most $z$ outliers, by $k$ balls of radius $\rho$,
and that $\dmin$ and $\dmax$ are bounds on the minimum and maximum distance between any two 
points in the stream. Let $\rho$ be a parameter with $\dmin \leq \rho \leq \dmax$.
We will design a sketch $\sketch(t)$ and a corresponding decision
algorithm \textsc{TryToCover} with the following properties:

\begin{itemize}
\item The sketch $\sketch(t)$ uses $O(kz/\eps^d)$ storage.
\item \textsc{TryToCover} either reports a collection $\C^*$ of $k$ balls 
      of radius $\optkz(P(t))+2\eps\rho$ that together cover all points in $P(t)$
      except for at most $z$ outliers, or it reports~\NO. In latter case we have $\optkz(P(t))> 2\rho$.
\end{itemize}

\noindent Our sketch $\sketch(t)$ is a tuple $[\ts(t),\B(t), \R(t)),\Pout(t)]$,
where $\ts$ is a \emph{timestamp}, $\B(t)$ is a collection of mini-balls, 
$\R(t)$ is a collection of so-called \emph{representative sets},
and $\Pout(t)$ is a set of outliers. We will maintain the following invariants.
\\
\begin{description}
\item[(Inv-1)] For any time $t'$ with
      $t\leq t'<\ts(t)$ we are guaranteed that $\optkz(P(t'))> 2 \rho$. The idea is that
      when we find $k+z+1$ points with pairwise distances greater~$2\rho$,
      then we can set $\tau$ to the expiration time of the oldest of these points,
      and we can delete this point (as well as any older points) from the sketch.
      % $2\rho$ because of how we use it for the optimization version of the problem
\item[(Inv-2)] 
      Each mini-ball $B\in\B(t)$ has radius~$\eps \rho$ and the center of $B$,
      denoted by $\cent(B)$, is a point from the stream.
      Note that the center does not need to be a point from $P(t)$.
      The mini-balls in $\B(t)$ are \emph{well spread} in the sense that 
      no mini-ball contains the center of any other mini-ball. In other words, we have
      $\dist(\cent(B),\cent(B'))> \eps\rho$ for any two mini-balls $B,B'\in\B(t)$.
\item[(Inv-3)] For each mini-ball $B\in \B(t)$ the set $\R(t)$ contains a \emph{representative set} 
      $R(B) \subseteq B \cap P(t)$, and these are the only sets in $\R(t)$.
      The representative sets $R(B)$ are pairwise disjoint, and each set $R(B)$
      contains at most~$z+1$ points. When $|R(B)|=z+1$, we say that the mini-ball~$B$ is \emph{full}.
\item[(Inv-4)]       
      Define $\Pin(t) := \bigcup_{B\in\B(t)} R(B)$ and let $S(t) := \Pin(t) \cup \Pout(t)$
      be the collection of all points in our sketch. 
      For any point $p_i \in P(t)\setminus S(t)$---these are exactly the points that 
      have not yet expired but that have been discarded by the algorithm---we have 
       (i) $\et(p_i) \leq \ts(t)$, and/or 
       (ii) $p_i\in B$ for a mini-ball $B\in\B(t)$ that is full and such that all points
            in $R(B)$ arrived after~$p_i$.
\item[(Inv-5)] $|\B(t)|=O(k/\eps^d)$ and $|\Pout(t)|\leq z$. \\
\end{description}
At time $t=0$, before the arrival of the first point, we have $\ts(t)=0$ and $\B(t)=\R(t)=\Pout(t)=\emptyset$.
Since $P(0)=\emptyset$ this trivially satisfies the invariants.
Before we prove that we can maintain the sketch upon the arrival and departure of points, 
we present our decision algorithm \textsc{TryToCover} and prove its correctness. 
\pagebreak
%The algorithm is quite simple, and given in Algorithm~\ref{alg:TryToCover}.
%------------------------------------------------------------------------------------------
\begin{algorithm}[h] 
\caption{\textsc{TryToCover}$(\sketch(t))$} \label{alg:TryToCover}
\begin{algorithmic}[1]
\State $S(t) \gets \bigcup_{B\in\B(t)} R(B)$ 
\State Compute $\optkz(S(t))$ and the corresponding collection $\C := \{C_1,\ldots,C_k\}$ of balls. \label{li:compute-opt}
\If{$t < \ts(t)$ or $\optkz(S(t)) > 2\rho$}
    \State Report \NO
\Else
    \State Increase the radius of each ball $C_i\in\C$ by $2\eps\rho$. \label{li:expand}
    \State Report the collection $\C^* := \{ C^*_1,\ldots,C^*_k\}$ of expanded balls. \label{li:report-expanded}
\EndIf
\end{algorithmic}
\end{algorithm}
%------------------------------------------------------------------------------------------

In line~\ref{li:compute-opt} we compute an optimal solution on the point set~$S(t)$
How this can be done depends on the underlying metric space and will be discussed in Section~\ref{subse:report-solution}.
The following lemma establishes the correctness of the algorithm.
%------------------------------------------------------------------------------------------
\begin{lemma} \label{le:correctness}
Algorithm {\rm \textsc{TryToCover}} either reports a collection $\C^*$ of $k$ balls 
of radius $\optkz(P(t))+2\eps\rho$ that together cover all points in $P(t)$ except
for at most $z$ outliers, or it reports~\NO. 
In latter case we have $\optkz(P(t))> 2\rho$.
\end{lemma}
%------------------------------------------------------------------------------------------
\begin{proof}
First suppose the algorithm reports~\NO. If this happens because $t<\ts(t)$ then 
$\optkz(P(t))> 2\rho$ by~(Inv-1). Otherwise, this happens because $\optkz(S(t)) > 2\rho$.
But then $\optkz(P(t))> 2\rho$, because (Inv-3) implies that $S(t) \subseteq P(t)$.

Now suppose the algorithm reports a collection $\C^*:=\{C^*_1,\ldots,C^*_k\}$ of balls.
Let $\C$ be the corresponding set of balls before they were expanded.
Since $\C$ is an optimal solution for $S(t)$, the balls $C_i$ have radius 
$\optkz(S(t))\leq \optkz(P(t))$ and together they cover all points in~$S(t)$
except for at most $z$ outliers. 
Now consider a point $p_i\in P(t)\setminus S(t)$. To finish the proof, we must show
that $p_i$ is covered by one of the balls in~$\C^*$.
To this end, first observe that $\et(p_i)>t$ because $p_i\in P(t)$. Since \textsc{TryToCover}
did not report~\NO, this implies that $\et(p_i)>\ts(t)$. 
Hence, we can conclude from (Inv-4) that $p_i\in B$ for a mini-ball $B\in \B(t)$ that is full.
Thus $R(B)$ contains $z+1$ points, and since we allow only $z$ outliers this implies
that at least one point from $R(B)$ is covered by a ball $C_i\in \C$. Because
$\diam(B)=2\eps\rho$, this implies that $p_i$ must be covered by~$C^*_i$, thus finishing the
proof.
\end{proof}
%------------------------------------------------------------------------------------------
Next we show how to update the sketch~$\sketch(t)$.

%------------------------------------------------------------------------------------------
\paragraph*{Handling departures}
%------------------------------------------------------------------------------------------
When a point $p_j$ in one of our representative sets $R(B)$
expires, we simply delete it from $R(B)$, and if $R(B)$ then becomes empty we remove
$R(B)$ from $\R(t)$ and $B$ from $\B(t)$. Similarly, if $p_j$ was a point in $\Pout(t)$
we remove $p_j$ from $\Pout(t)$.

It is trivial to verify that (Inv-1)--(Inv-3) and (Inv-5)
still hold for the updated sketch. To see that (Inv-4) holds as well, consider a point
$p_i\in P(t)\setminus S(t)$. The only reason for (Inv-4) to be violated, would be when
$p_i\in B$ for a mini-ball $B$ that was full before the deletion of $p_j$ but is no longer
full after the deletion. However, (Inv-4) states that all points in $R(B)$ arrived
after~$p_i$. Since $p_i$ did not yet expire, this means that the point~$p_j$ that currently
expires cannot be a point from $R(B)$. 

%------------------------------------------------------------------------------------------
\paragraph*{Handling arrivals}
%------------------------------------------------------------------------------------------
The arrival of a point $p_j$ at time $t:= \at(p_j)$ is handled by Algorithm~\ref{alg:arrival}. 
We denote the sketch just before the arrival by $\sketch(t^-)$, and the updated sketch by $\sketch(t^+)$.

%------------------------------------------------------------------------------------------
\begin{algorithm}[h]
\caption{\textsc{HandleArrival}$(\sketch(t^-),p_j)$} \label{alg:arrival}
\vspace*{-2mm}
\begin{enumerate}
\item If $p_j$ lies in an existing mini-ball~$B\in \B(t^-)$ then add $p_j$ to $R(B)$.
      If $p_j$ did not lie in an existing mini-ball then add $p_j$ to $\Pout(t^-)$.
\item \label{step:charikar} Let $Q := \Pout(t^-) \cup \bigcup_{B\in \B(t^-)} R(B)$ be the set of all points
      in the current sketch, including the just inserted point. 
      Sort the points in $Q$ by decreasing arrival time, and let $q_1,q_2\ldots$
      denote the sorted sequence of points---note that $q_1$ is the newly inserted point~$p_j$.
      Let $Q_i := \{q_1,\ldots,q_i\}$ denote the set containing the first $i$ points from~$Q$.
      Find the largest index~$i^*$ such that the following holds:
      \begin{itemize}
      \item The 3-approximation algorithm of Charikar~\etal~\cite{charikar2001algorithms}
            for the $k$-center clustering algorithm with $z$ outliers, when run on $Q_{i^*}$,
            reports a collection $\B^*$ of $k$ balls of radius at most~$6\rho$.
            (In fact, any other constant-factor approximation algorithm can be used if we adjust the value $6\rho$ appropriately. For concreteness, we use the one by Charikar~\etal)
      \end{itemize}
\item \label{step:update-ts}
      If $Q_{i^*}\neq Q$ (and so the point~$q_{i^*+1}$  exists) then set
      $\ts(t^+) := \max(\ts(t^-,\et(q_{i^*+1}))$.
      Remove the points in $Q\setminus Q_{i^*}$ from the representative sets they are in;
      note that these points all have expiration time at least $\ts(t^+)$.
\item \label{step:new-balls}
      Let $\B^* := \{B_1,\ldots,B_k\}$ be the balls in the solution reported for $Q_{i^*}$.
      Note that all points from $Q_{i^*}$, except for at most $z$ outliers, lie in a ball from~$\B^*$. 
      The new set of mini-balls and outliers is now computed as follows.
      Increase the radius of the balls in $\B^*$ by~$\eps\rho$.
      Denote the set of expanded balls by $\overline{\B}^*$ 
      \begin{enumerate}
      \item \label{step:reused-balls} Add each mini-ball $B\in \B(t^-)$ whose center lies inside a ball
            from $\overline{\B}^*$ to the set $\B(t^+)$ and add their representative sets
           to~$\R(t^+)$.
      \item \label{step:create-balls}
            Let $Z\subset Q_{i^*}$ be the set of points that are not in a representative set added
            to $\R(t^+)$ in Step~\ref{step:reused-balls}. Let
            $Z_{\myin}\subseteq Z$ be the points that lie inside a ball from~$\overline{\B}^*$.
            Go over the points in $Z_{\myin}$ one by one, in arbitrary order, and handle
            each point~$q$ as follows: If there is a mini-ball in the current set~$\B(t^+)$ 
            that contains~$q$ then add $q$ to $R(B)$ for an arbitrary such mini-ball~$B$;
            otherwise create a mini-ball $B$ with center $q$ and radius~$\eps\rho$,
            set $R(B) := \{q\}$ and add $B$ to ~$\B(t^+)$.
      \item \label{step:remove-oldest} For all mini-balls $B\in \B(t^+)$ such that $|R(B)|>z+1$, 
            remove the oldest $|R(B)|-(z+1)$ points from $R(B)$.
      \item Set $\Pout(t^+) := Z\setminus Z_{\myin}$. Note that these are exactly 
             the points from $Q_{i^*}$ that have not been added to a representative set
      \end{enumerate}
\end{enumerate}
\end{algorithm}
%------------------------------------------------------------------------------------------
% We now show that the invariants are maintained.
%--------------------------------------------------------------------------------------
\begin{lemma}
After the arrival of point $p_j$ has been handled, the invariants (Inv-1)--(Inv-5) are restored.
\end{lemma}
%--------------------------------------------------------------------------------------
\begin{proof}
\begin{description}
\item[(Inv-1)]
    If $\ts(t^+) = \ts(t^-)$ then obviously (Inv-1) still holds, so assume that $\ts$ is
    updated by the algorithm in Step~\ref{step:update-ts}. Thus the 
    algorithm of Charikar~\etal~\cite{charikar2001algorithms} returned a solution
    of radius more than~$6\rho$ on $Q_{i^*+1}$. Since the algorithm is a 3-approximation,
    this means $\optkz(Q_{i^*+1})>2\rho$. Because we sorted the points on expiration time,
    we must have $\optkz(P(t'))> 2 \rho$ for all $t'$ with $t\leq t' < \et(q_{i^*+1})$. 
    Hence, (Inv-1) still holds.
\item[(Inv-2)]
    By construction, the mini-balls in $\B(t^+)$ have radius~$\eps\rho$ and are centered at a point 
    from the stream. The mini-balls from $\B(t^-)$ are well spread by induction, so we only
    need to prove that any mini-balls created in Step~\ref{step:create-balls}
    are well spread. This is true because we only create a new mini-ball when its center does not
    lie inside an existing mini-ball.
\item[(Inv-3)]
    The representative sets~$R(B)$  are disjoint by construction, and $|R(B)|\leq z+1$
    is guaranteed by Step~\ref{step:remove-oldest}. Hence (Inv-3) holds.
\item[(Inv-4)]
    Take a point~$p_i\in P(t^+)\setminus S(t^+)$ and assume $\et(p_i) > \ts(t^+)$. 
    We have two cases.
    \begin{itemize}
    \item The first case is that $p_i$ was discarded due to the arrival of $p_j$. 
          Then it was either removed in Step~\ref{step:update-ts}, which implies 
          $\et(p_i) \geq \ts(t^+)$; or it was removed in Step~\ref{step:remove-oldest},
          which implies there is a full mini-ball $B\in\B(t^+)$ such that all points in $R(B)$ are newer
          than~$p_i$. Hence, (Inv-4) holds for~$p_i$
    \item The second case is that $p_i$ was already discarded earlier. We can assume that
          (Inv-4) holds before the arrival of~$p_j$. If $\et(p_i) \leq \ts(t^-)$
          then we also have $\et(p_i) \leq \ts(t^+)$ and we are done. So assume that 
          $p_i\in B$ for a mini-ball $B\in\B(t^-)$ that was full and such that all points
          in $R(B)$ arrived after~$p_i$. If~$R(B)$ contained a point $p_{i'}\in Q\setminus Q_{i^*}$
          then $\et(p_i) < \et(p_{i'}) \leq \ts(t^+)$, and so (Inv-4) holds. If~$R(B)$ 
          only contained points from $Q_{i^*}$ then we have two sub-cases.
          
          The first sub-case is that $B$ is still present in $\B(t^+)$. Then $B$ must still
          be full, because $R(B)$ only contained points from~$Q_{i^*}$.
          Hence, (Inv-4) holds. 
          
          The second sub-case is that $B$ is not present in~$\B(t^+)$. Thus its
          center must lie outside all balls in~$\overline{\B}^*$. But then the points
          in $R(B)$, which are all in $Q_{i^*}$, are outside any ball in~$\B^*$. But this is 
          impossible, since $|R(B)|=z+1$ and $\B^*$ has only~$z$ outliers.
    \end{itemize}
\item[(Inv-5)]    
    All mini-balls in $\B(t^+)$ have their center inside one of the $k$ balls in $\overline{\B}^*$,
    which have radius at most~$6\rho$. Moreover, the mini-balls in $\B(t^+)$ have radius~$\eps\rho$
    and they are well spread. Since the underlying space has doubling dimension~$d$,
    this implies that $|\B(t^+)| =O(k/\eps^d)$.
    
    The outlier set $\Pout(t^-)$
    consists of the points from $Q_{i^*}$ that lie outside the balls in~$\overline{\B}^*$.
    Hence, these points also lie outside the balls in~$\B^*$. Since $\B^*$ is a valid solution to the
    $k$-center problem with $z$ outliers, there are at most $z$ such points.
    Thus $|\Pout(t^+)|\leq z$.
\end{description}
\end{proof}
%------------------------------------------------------------------------------------------

%------------------------------------------------------------------------------------------
\subsection{A sketch for the optimization problem}
%------------------------------------------------------------------------------------------
Above we presented a sketch for a decision version of the problem, for given 
parameters $\rho$ and~$\eps$. The sketch uses $O(kz/\eps^d)$ storage.
We also gave an algorithm\textsc{TryToCover} that either reports a 
collection $\C^*$ of $k$ balls of radius $\optkz(P(t))+2\eps\rho$ covering all
points in $P(t)$ except at most $z$ outliers, or that reports~\NO. In latter case we know that $\optkz(P(t))> 2\rho$.
To make the parameter $\rho$ and $\eps$ explicit, we will from now on denote the sketch
by $\sketch_{\rho,\eps}$.

In the optimization version of the problem we wish to find $k$ congruent balls of minimum radius
that together cover all points in $P(t)$ except for at most $z$ outliers. To obtain a
$(1+\eps)$-approximation for the optimization problem, for a given $\eps>0$,
we maintain a sketch $\sketch_{\rho_i,\eps/2}$ for every
$0\leq i \leq \floor{\log \sigma}$, where $\rho_i := 2^i \cdot \dmin$.
We let $\sketch(t) := \{ \sketch_{\rho_0,\eps/2},\ldots, \sketch_{\rho_{\floor{\log \sigma}},\eps/2} \}$ denote the collection of these sketches at time~$t$.
We can then obtain a $(1+\eps)$-approximate solution with the algorithm
shown in Algorithm~\ref{alg:optimize}.
%------------------------------------------------------------------------------------------
\begin{algorithm}[h] 
\caption{\textsc{FindApproximateCenters}$(\sketch(t))$} \label{alg:optimize}
\begin{algorithmic}[1]
\State $i\gets 0$
\Repeat  \label{li:start-repeat}
    \State $\answer \gets \textsc{TryToCover}(\sketch_{\rho_i,\eps/2})$;  $i\gets i+1$
\Until{$\answer \neq$ \NO} \label{li:end-repeat}
\If{the radii of the balls in $\answer$ is less than $\dmin$}
    \State Reduce the radii of the balls to zero
\EndIf
    \State Report $\answer$
\end{algorithmic}
\end{algorithm}
%------------------------------------------------------------------------------------------

\noindent We obtain the following theorem. The theorem implicitly assumes that algorithm
\textsc{tryToCover}, which is called by \textsc{FindApproximateCenters}, has access to
a subroutine for computing an optimal solution. This issue will be further discussed in Section~\ref{subse:report-solution}.
%------------------------------------------------------------------------------------------
\begin{theorem} \label{th:main}
The sketch $\sketch(t)$ uses $O((kz/\eps^d) \log\sigma)$ storage, and
\textsc{FindApproximateCenters}$(\sketch(t))$ reports a collection $\C^*$ of balls of
radius at most $(1+\eps) \cdot \optkz(P(t))$ covering all points from~$P(t)$, except at most~$z$ outliers. 
\end{theorem}
%------------------------------------------------------------------------------------------
\begin{proof}
The storage bound follows immediately from the fact that each of the
sketches $\sketch_{\rho_i,\eps/2}$ uses $O((kz/\eps^d))$ storage.
It remains to prove the second part of the theorem.

First consider the case $\optkz(P(t))=0$. When we run $\textsc{TryToCover}(\sketch_{\rho_i,\eps/2})$
with $i=0$, then by Lemma~\ref{le:correctness} we will report a collection of balls of radius 
\[
\optkz(P(t)) + 2 (\eps/2)\rho_0 = \eps \cdot \dmin.
\]
Since $\eps<1$, the radii are strictly smaller than 
$\dmin$. Since the balls are centered at points
from $X$, we may as well reduce the radii to zero, thus achieving 
an optimal solution.

Next, consider the case $\optkz(P(t))>0$. Then $\optkz(P(t))\geq \dmin=\rho_0$. Consider the loop in lines~\ref{li:start-repeat}--\ref{li:end-repeat}
and let $i^*$ be the value of the counter~$i$ when we obtain
$\textsc{TryToCover}(\sketch_{\rho_i,\eps/2}) \neq$~\NO. Note that this
must happen at some point, because for $i = \floor{\log \sigma}$
we have
\[
\rho_i = 2^{\floor{\log \sigma}} \cdot \dmin 
    \geq (\sigma/2) \cdot \dmin 
    = (1/2) \cdot \dmax
\]
and so Lemma~\ref{le:correctness} guarantees that $\answer\neq$~\NO.
If $i^*=0$ then the balls in the reported solution have radius
\[
\optkz(P(t)) + 2 (\eps/2)\rho_0  \leq (1+\eps)\cdot \optkz(P(t)).
\]
Otherwise we know that the answer for $i=i^*-1$ is~\NO, which implies that
$\optkz(P(t))>2\rho_{i^*-1}=\rho_{i^*}$. Hence, the balls in the solution that is
reported for $i^*$ have radius
\[
\optkz(P(t)) + 2 (\eps/2)\rho_{i^*}  \leq (1+\eps)\cdot \optkz(P(t)).
\]
\end{proof}
%------------------------------------------------------------------------------------------
Our sketch for the $k$-center problem can also be used
for the $k'$-center problem for $k'<k$. Moreover, a sketch for the $k$-center problem for~$k\geq 1$
can be used for the diameter problem. Recall that the diameter problem with outliers for the set $P(t)$ asks for the value
\[
\diamz(P(t)) := \min \{ \diam(P(t)\setminus Q) : |Q|=z \}, 
\]
that is, $\diamz(P(t))$ is the smallest diameter one can obtain by deleting $z$
outliers from~$P(t)$. We say that an algorithm reports a $(1-\eps)$-approximation
to $\diamz(P(t))$ if it reports a value~$D$ with $(1-\eps)\cdot \diamz(P(t)) \leq D \leq \diamz(P(t))$.
%------------------------------------------------------------------------------------------
\begin{theorem} \label{th:diameter}
The sketch for the $k$-center problem with outliers as presented above can also be used to 
provide a $(1+\eps)$-approximation for the $k'$-center problem with outliers, for any $1\leq k'\leq k$.
Moreover, it can be used to provide a $(1-2\eps)$-approximation for the diameter problem with
outliers. 
\end{theorem}
%------------------------------------------------------------------------------------------
\begin{proof}
The value $k$ plays a role in (Inv-1) and (Inv-5), not in the other invariants.
In (Inv-5) the value of $k$ determines the size of $\B(t)$---it has no influence on
the correctness. As for (Inv-1), we observe that $\opt_{k',z}(P(t)) \geq \optkz(P(t))$ for
any $k'\leq k$, so the fact that (Inv-1) holds for $k$ implies that it holds
for $k'<k$ as well. Hence, a sketch for the $k$-center problem will have the properties
required of a sketch for the $k'$-center problem (except that the storage depends on $k$
instead of~$k'$.) Thus running \textsc{TryToCover} on a
sketch for the $k$-center, where in line~\ref{li:compute-opt} we compute $\opt_{k',z}(P(t))$,
will give a correct result for~$k'$.
\medskip

Now consider the diameter problem. Suppose we run \textsc{TryToCover} on a
sketch for the $k$-center problem, where in line~\ref{li:compute-opt} we compute $\diamz(P(t))$,
and instead of lines~\ref{li:expand}--\ref{li:report-expanded} we report~$D :=\diamz(S(t))$.
We claim that if the algorithm reports~\NO then we have $\diamz(P(t)) > 2\rho$, and otherwise $\diamz(P(t))-4\eps\rho \leq D \leq \diamz(P(t))$.

Note that $\diamz(P(t)) \geq \optkz(P(t))$ for any $k\geq 1$. Hence, when $t<\ts(t)$
then we have  $\diamz(P(t)) \geq \optkz(P(t)) > 2\rho$.
This implies that the claim holds when \textsc{TryToCover} reports~\NO.

If \textsc{TryToCover} does not report~\NO, it reports $D:=\diamz(S(t))$.
Clearly we then have $D \leq \diamz(P(t))$. Now suppose for a contradiction that $D<\diamz(P(t))-4\eps\rho$. Let $p_i,p_j\in P(t)$ be such that
$\dist(p_i,p_j) = \diamz(P(t))$. We will argue that there are points
$p_{i'},p_{j'}\in S(t)$ such that $\dist(p_i,p_{i'})\leq 2\eps\rho$ and
$\dist(p_j,p_{j'}) \leq 2\eps\rho$. But then we would have $D=\diamz(S(t))\geq \diamz(P(t))-4\eps\rho$, which contradicts the assumption.
We will argue the existence of $p_{i'}$; the argument for $p_{j'}$ is similar.
If $p_i\in S(t)$ then we can take $p_{i'}:= p_i$ and we are done. Otherwise, as argued in the proof of Lemma~\ref{le:correctness}, we know that $p_i\in B$ for a ball $B\in B(t)$ that is full. In particular $R(B)$
contains at least one point $p_{i'}$, and this point must be at distance at most
$2\eps\rho$ from~$p_i$, as claimed.

We have proved the claim that if the algorithm reports~\NO then we have $\diamz(P(t)) > 2\rho$, and otherwise we have $\diamz(P(t))-4\eps\rho \leq D \leq \diamz(P(t))$. Plugging this claim for the decision problem into the mechanism 
to obtain a sketch for the optimization problem---recall that there we used $\eps/2$ as parameter---now gives the desired result.
\end{proof}
%------------------------------------------------------------------------------------------
%------------------------------------------------------------------------------------------
Theorem~\ref{th:diameter} implies that there exists a sketch for the diameter
problem with outliers in the sliding-window model that gives a $(1+\eps)$-approximation 
to $\diamz(P(t))$ using $O((z/\eps^d)\log\sigma)$ storage, namely the sketch
for the 1-center problem.

%------------------------------------------------------------------------------------------
\subsection{Time complexity}
\label{subse:report-solution}
%------------------------------------------------------------------------------------------
Above we focused on the storage used by our sketch, and on the approximation ratio it can potentially
provide. We now discuss the time complexity.

The algorithm that handles departures trivially runs in $O(1)$ time. The most time-consuming
step in the algorithm that handles arrivals is Step~\ref{step:charikar}, where we run the algorithm of 
Charikar~\etal~\cite{charikar2001algorithms} on a set of $O(kz/\eps^d)$ points. This takes
$O((kz/\eps^d)^3)$ time. 
%
% I THINK THIS CAN BE IMPROVED< BUT THE PAPER SAYS O(n^3)
%
Note that, both for departures and arrivals, we need
to execute the respective algorithms for each of the $\log\sigma$ sketches~$\sketch_{\rho_i,\eps/2}$
we maintain. 

The time needed for \textsc{FindApproximateCenters} is more interesting, since it calls 
{\sc TryToCover}, which is supposed to compute an optimal solution for the $k$-center problem
with outliers on a given set $S(t)$ of $O(kz/\eps)$ points. It is routine to check that if
{\sc TryToCover} would use a $c$-approximate solution on~$S(t)$, instead of an optimal solution,
our final result would be a~$c(1+\eps)$-approximation. 
In $d$-dimensional Euclidean space  we can solve the $k$-center
problem with outliers in time polynomial in $n:=|S(t)|$  (for constant $k$ and $d$),
as follows: first generate all $O(n^{d+1})$ potential centers---there are $O(n^{d+1})$ potential
centers because the smallest enclosing ball of a given point set is defined
by at most $d+1$ points---then generate all 
possible collections of $k$ such centers, and then for each of the $O(n^{(d+1)k})$ 
such collections find the minimum radius~$\rho$ such that we can cover all except for 
$z$ points. A similar approach is possible for other metrics in $\Reals^d$
that are sufficiently well-behaved (such as $\ell_p$ metrics, for instance).

For arbitrary metric spaces the situation is more tricky. The standard assumption,
which we also make, is that we can compute the distance between any two given points
from the metric space in $O(1)$ time. Still, computing an optimal solution on the
set $S(t)$ can be quite slow or perhaps even infeasible. For example,
when an optimal center is a point from the metric space that does not occur in
the stream, then the algorithm may be unable to retrieve this center. 
So unless we are in Euclidean space, or some other ``well-defined'' space, it 
seems most natural to define the underlying space as consisting of exactly the points
in the stream. Note that this can still be somewhat problematic, as it may require the
algorithm to use points that have already expired as centers for the current set~$S(t)$.
These issues seem unavoidable for any algorithm that maintains a sketch in the sliding-window
model. Also note that, using just the points in the current set $P(t)$
as centers, one can always obtain a 2-approximation to the optimal clustering. 
Invariant~(Inv-4) thus implies that by using the set~$S(t)$ 
as potential centers we can obtain a $(2+\eps)$ approximation. Also note that
the optimal $k$-center clustering with $z$ outliers on $S(t)$, under the restriction
that the centers come from $S(t)$, can be trivially computed in $O(|S(t)|^{k+2})$ time.

The following theorem summarizes the performance of our sketch.
%------------------------------------------------------------------------------------------
\begin{theorem} \label{th:main-plus-alg}
Consider the $k$-center problem with $z$ outliers in the sliding-window model, for
streams of points from a space with doubling dimension~$d$. Let $\dmin$ and $\dmax$ be given 
(lower and upper) bounds on the minimum and maximum distance between any two points in the
stream, and let $\sigma := \dmax/\dmin$ be (an upper bound on) the spread of the stream.
Suppose we have a $c$-approximation algorithm for the static version of the problem that runs in $T_{k,z}(n)$ time
on a set of $n$ points. Then for any $0<\eps<1$ we can maintain a sketch with the following properties:
\begin{itemize}
\item The sketch uses $O((kz/\eps^d) \log\sigma)$ storage.
\item Departures and arrivals can be handled in $O(\log\sigma)$ and $O((kz/\eps^d)^3\log\sigma)$ time, respectively.
\item At any time $t$, the algorithm can report a valid solution for $P(t)$ of cost
     at most $(1+\eps)c \cdot \optkz(P(t))$.
      The time needed to compute the solution is~$T_{k,z}(n)$, where $n:= O(kz/\eps^d)$.
\end{itemize}
\end{theorem}
%------------------------------------------------------------------------------------------

%------------------------------------------------------------------------------------------
\section{A lower bound}
\label{se:lower-bound}
%------------------------------------------------------------------------------------------
Above we presented a sketch of size $O((kz/\eps^d)\log \sigma)$ that provides 
a $(1+\eps)$-approximation for the $k$-center problem with $z$ outliers in the sliding-window model. 
% where $d$ is the doubling dimension of the underlying space and $\sigma$ is its spread.
In this section we show that this is tight, up to the dependency on~$d$.
We will first show that to obtain 
any constant approximation ratio for the problem in~$\Reals^1$, a sketch must use~$\Omega(kz)$ storage.
Then we prove that, for $0<\eps<1$, any sketch giving a $(1+\eps)$-approximation
in~$\Reals^1$ must use $\Omega((kz/\eps)\log \sigma)$ storage.

%------------------------------------------------------------------------------------------
\paragraph*{The lower-bound model}
%------------------------------------------------------------------------------------------
Our lower-bound model is extremely simple and general. We allow the algorithm to store points,
or weighted points, or balls, or whatever it wants so that it can approximate an optimal
solution for $P(t)$ at any time~$t$. We only make the following assumptions. Let $S(t)$
be the collection of objects being stored at time~$t$.
\begin{itemize}
\item Each object in $S(t)$ is accompanied by an expiration time, which is equal to the expiration 
      time of some point $p_i\in P(t)$. 
\item Let $p_i\in P(t)$. If no object in $S(t)$ uses $\et(p_i)$
      as its expiration time, then no object in $S(t')$ with $t'>t$ can use $\et(p_i)$
      as its expiration time. (Once an expiration time has been discarded, it cannot be recovered.)
\item The solution reported by the algorithm is uniquely determined by $S(t)$, and
      the algorithm only modifies $S(t)$ when a new point arrives or when an object in~$S(t)$ expires.
\item The algorithm is deterministic and oblivious of future arrivals. In other words, the set
      $S(t)$ is uniquely determined by the sequence of arrivals up to time~$t$, and the
      solution reported for $P(t)$ is uniquely determined by~$S(t)$.
\end{itemize}
The storage used by the algorithm is defined as the number of objects in~$S(t)$.
The algorithm can decide which objects to keep in $S(t)$ anyway it wants; it may even 
keep an unbounded amount of extra information in order to make its decisions. The algorithm can 
also derive a solution for $P(t)$ in any way it wants, as long as the solution is valid and uniquely
determined by~$S(t)$. Clearly, the sketch from the previous section adheres to the model.
\medskip

%------------------------------------------------------------------------------------------
\paragraph*{A lower bound for constant-factor approximations}
%------------------------------------------------------------------------------------------
Next we give a simple example showing that
any constant-factor approximation must use $\Omega(kz)$ storage. 

Let \alg be an algorithm that maintains a $c$-approximation for the problem in~$\Reals^1$. 
Consider the set $P(t)$ shown in Figure~\ref{fi:lower-bound-constant}.
The set $P(t)$ consists of $k+1$ clusters $C_1,\ldots,C_{k+1}$ 
of $z+1$ points each. The points within a cluster are at distance~1
from each other (so the radius of a cluster is~$z/2$) and the distance between two
consecutive clusters is~$cz+1$. Note that the spread of the set is $(k+1)z+k(cz+1) =\Theta(ckz)$. 
%------------------------------------------------------------------------------------------
\begin{figure}
\begin{center}
\includegraphics{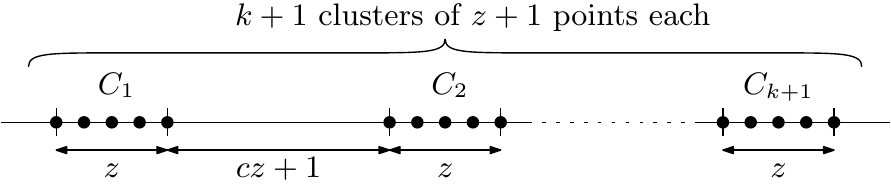}
\end{center}
\caption{The lower-bound construction showing that any $c$-approximation algorithm must 
          use $\Omega(kz)$ storage.}
\label{fi:lower-bound-constant}
\end{figure}
%------------------------------------------------------------------------------------------
Let $\Texp$ be the set of expiration times of the points in $P(t)$, except for the first point.
Note that $|\Texp|=(k+1)(z+1)-1$. Suppose at least
one of these expiration times, say the expiration time of point~$p^*$,
is not used by any of the objects in~$S(t)$. The adversary now proceeds as follows.
Let $Q\subset P(t)$ be the set of points expiring before~$t_{p^*}$, where $t_{p^*} := \et(p^*)$
is the expiration time of~$p^*$. Whenever a point in $Q$ expires, the adversary immediately replaces 
it by a point at the same location. 

Now let $t_{p^*}^-$ and $t_{p^*}^+$ be the times immediately before
and after~$t_{p^*}$, respectively. Then at time $t_{p^*}^-$ we have the point set shown in
Figure~\ref{fi:lower-bound-constant}, while at time $t_{p^*}^+$ the point~$p^*$ has expired.
Since we are allowed to use only $k$ balls and $z$ outliers, any valid solution at
time $t_{p^*}^-$ must have a ball that contains points from at least two of the clusters. 
Since the distance between any two clusters is at least $cz+1$, the radius of such a ball is at least $(cz+1)/2$.
Hence, $\optkz(P(t_{p^*}^-))\geq (cz+1)/2$. 
On the other hand, at time $t_{p^*}^+$ there is a cluster~$C_{i^*}$ of size $z$.
Thus \alg can designate the points from $C_{i^*}$ as outliers and cover the remaining points 
with $k$ balls of radius at most $z/2$, which implies $\optkz(P(t_{p^*}^+)) \leq z/2$.

Since the algorithm must report a valid solution at time~$t_{p^*}^-$ and it
must report the same solution at time~$t_{p^*}^+$, the approximation ratio 
at time~$t_{p^*}^+$ is at least
\[
\frac{\optkz(P(t^-_{p^*}))}{\optkz(P(t^+_{p^*}))} 
    \geq \frac{(cz+1)/2}{z/2} > c
\]
which contradicts that \alg is a $c$-approximation algorithm.
We obtain the following theorem.
%------------------------------------------------------------------------------------------
\begin{theorem}\label{th:lower-bound-constant}
Let \alg be a $c$-approximation algorithm for the $k$-center problem with
$z$ outliers in $\Reals^1$, with $z\geq 1$, that works in the model described above.
Then there is a problem instance of spread $\Theta(kcz)$ that requires \alg to 
use $\Omega(kz)$ storage.
\end{theorem}
%------------------------------------------------------------------------------------------

%------------------------------------------------------------------------------------------
\paragraph*{A lower bound for a $(1+\eps)$-approximate sketch}
%------------------------------------------------------------------------------------------
Next we present a more intricate construction that shows that a logarithmic
dependence (in the storage bound) on the spread is necessary. It also shows that
to obtain a $(1+\eps)$-approximation, the storage must depend at least linearly
on~$1/\eps$.

Let \alg be a $(1+\eps')$-approximation algorithm, for some $0<\eps'<1$,
where we assume for simplicity that $1/\eps'$ is an integer. 
(The reason we work with $\eps'$ instead of $\eps$ will become clear shortly.)
The lower-bound instance, which consists of points in~$\Reals^1$, is as follows. 

%------------------------------------------------------------------------------------------
\begin{figure}
\begin{center}
\includegraphics[scale=0.9]{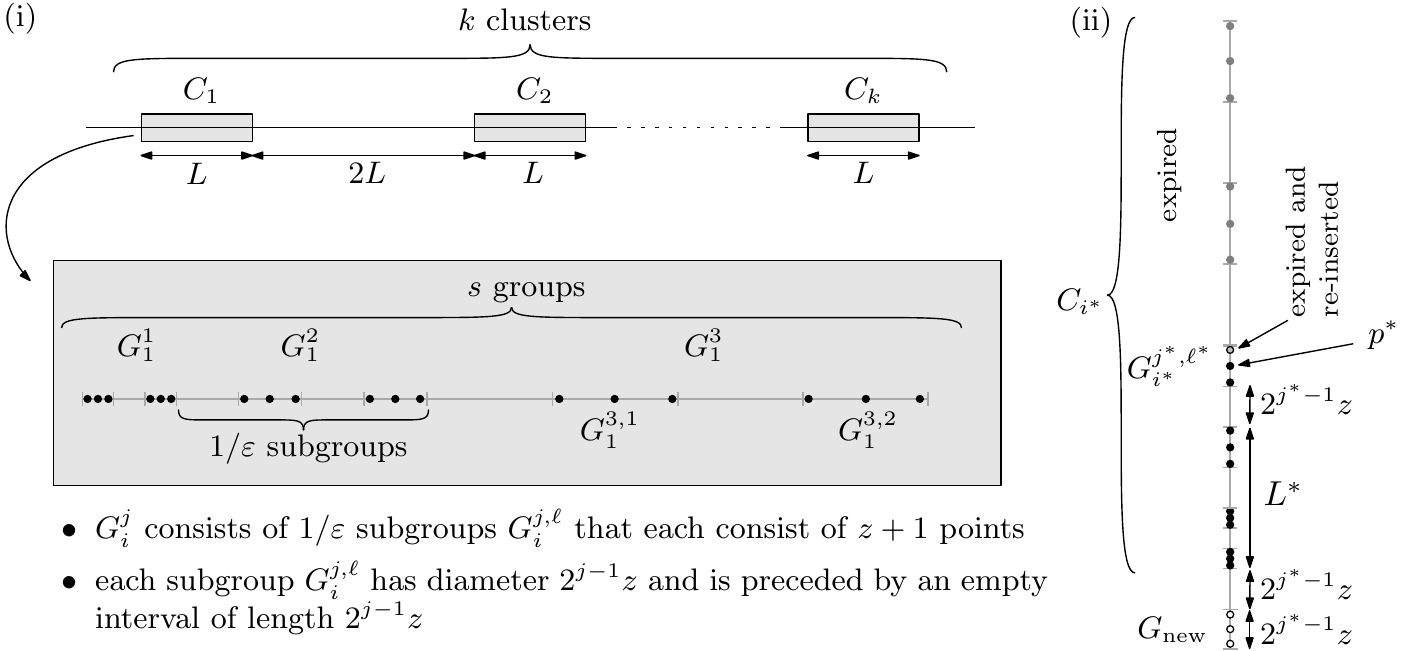}
\end{center}
\caption{The lower-bound construction for a $(1+\eps)$-approximate sketch. In the example, $s=3$,
         $k=3$, $\eps=1/2$ and $z=2$. (i) The initial configuration.
         (ii) Rotated view of the cluster containing~$p^*$.}
\label{fi:lower-bound-eps}
\end{figure}
%------------------------------------------------------------------------------------------
Let $\eps:=\eps'/8$ and consider the configuration shown in Figure~\ref{fi:lower-bound-eps}(i).
(This is the reason for starting with~$\eps'$: it allows us to
define $\eps := \eps'/8$ and describe the construction using $\eps$,
so we do not have to $\eps'$ all the time.)
The configuration consists of $k$ clusters $C_1,\ldots,C_k$, placed next to each other
from left to right. Each cluster $C_i$ consists of
$s$~groups, $G_{i}^1,\ldots,G_{i}^s$, where $s=\Theta(\log\sigma)$. Each group $G_{i}^j$
consists of $1/\eps$ subgroups, $G_{i}^{j,1},\ldots,G_{i}^{j,1/\eps}$. A subgroup~$G_{i}^{j,\ell}$
consists of $z+1$ points at distance~$2^{j-1}$ apart; the diameter of a subgroup---that
is, the distance between its leftmost and rightmost point---is thus~$2^{j-1}z$.
Subgroup~$G_{i}^{j,\ell}$ is preceded by an empty interval of length~$2^{j-1}z$.
Hence, the total diameter of the group $G_i$, including the empty interval
preceding its leftmost subgroup, is $2^{i-1}z/\eps$. (The exception is $G_1$,
for which we do not count the empty interval preceding its leftmost subgroup.)
This brings the total diameter of a cluster to 
\[
L := \sum_{j=1}^{s} 2^j z/\eps - z/\eps = (2^{s+1}-3)\cdot z/\eps < 2^{s+1}z/\eps.
\]
In between every two consecutive clusters there an an empty interval of length~$2L$.

The arrival order of the points in the configuration is as follows.
The subgroups within a cluster arrive from right to left, in a round-robin fashion
over the clusters: in the first round the subgroups 
$G_{k}^{s,\frac{1}{\eps}},\ldots,G_{1}^{s,\frac{1}{\eps}}$ arrive,
in the second round the subgroups $G_{k}^{s,\frac{1}{\eps}-1},\ldots,G_{1}^{s,\frac{1}{\eps}-1}$ arrive, and so on.
More formally, $G_{i}^{j,\ell}$ arrives before $G_{i'}^{j',\ell'}$ iff:
$j>j'$ or $(j=j' \mbox{ and } \ell>\ell')$ or $(j=j' \mbox{ and } \ell = \ell' \mbox{ and } i>i')$.
This finishes the description of the initial part of the instance. 
\medskip

Let $t$ be the time at which the last point in the configuration arrives, let $P(t)$
be the set of all points in the clusters~$C_1,\ldots,C_k$, and let $S(t)$
be the set of objects stored by~\alg. We will argue that $|S(t)| =\Omega(kzs/\eps)$, otherwise
an adversary can continue the instance such that at some time~$t'$
in the future the algorithm does not give a $(1+\eps)$-approximation.

Consider the points in the subgroups $G_{i}^{j,\ell}$ with $j>1$, 
except for the very first point that arrived. Let $\Texp$ be the set of expiration
times of these points. Note that $|\Texp|=(k(z+1)(s-1)/\eps) - 1$. Suppose at least
one of these expiration times, say the expiration time of point~$p^*$,
is not used by any object in~$S(t)$. Let 
$t_{p^*}^-$ and $t_{p^*}^+$ be the times immediately before and after~$\et(p^*)$, respectively.
The conditions of the model imply that the algorithm reports the same solution at times~$t_{p^*}^-$
and~$t_{p^*}^+$. The adversary will now continue the scenario in
such that one of these answers is not sufficiently accurate, as described next.

Let $G_{i^*}^{j^*,\ell^*}$ be the subgroup that $p^*$ belongs to.
First, the adversary waits until all points that arrived before $p^*$
have expired. The situation is then as follows:
clusters~$C_i$ with $i< i^*$ consist of the subgroups up to $G_{i}^{j^*,\ell^*}$,
while clusters $C_i$ with $i> i^*$ are missing the last of these subgroups.
Cluster $C_{i^*}$ has the points up to point~$p^*$ (which expires next) in $G_{i^*}^{j^*,\ell^*}$.
Before $p^*$ expires, the  adversary now adds a new subgroup $G_{\mathrm{new}}$ of $z+1$ points to $C_i$.
The diameter of $G_{\mathrm{new}}$ is $2^{j^*-1}z$ and its distance to~$C_{i^*}$ is $2^{j^*-1}z$ 
as well; see Figure~\ref{fi:lower-bound-eps}(ii).  (The construction
has been rotated by 90~degrees in the figure, to make it fit.)
In addition, the adversary adds the points from  $G_{i^*}^{j^*,\ell^*}$
that have already expired back. 
\medskip

We now analyze $\optkz(P(t^-_{p^*}))$. Since the distance between
any two of the $k$ clusters is larger than the diameter of the clusters, 
and any cluster still contains at least $z+1$ points---the latter follows because $j^*>1$---an
optimal solution will use a separate ball for each cluster. The largest cluster is
$C_{i^*}$, because $G_{\mathrm{new}}$ was added to it. Since each subgroup in $C_{i^*}$ has $z+1$
points and we can designate only $z$ points as outliers, it is optimal to designate
the $z$ rightmost points (topmost in Figure~\ref{fi:lower-bound-eps}(ii)) 
from $G_{i^*}^{j^*,\ell^*}$ as outliers and cover the remaining
points with a ball of diameter $L^* + 3\cdot 2^{j^*-1}z$,
where $L^*$ denotes the diameter of the union of the subgroups $G_{i^*}^{1,1}$ up to but excluding $G_{i^*}^{j^*,\ell^*}$; see Figure~\ref{fi:lower-bound-eps}.
Here the term $3\cdot 2^{j^*-1}z$ is needed because we must cover $G_{\mathrm{new}}$ as well as the leftmost (bottom-most in the figure) point from~$G_{i^*}^{j^*,\ell^*}$, which has not been designated as outlier. 
We have
\[
L^* := \sum_{j=1}^{j^*-1} 2^j z/\eps -z/\eps +  2\ell^*\cdot 2^{j^*-1} z 
     =  \left( 2^{j^*}-3 \right) z/\eps + \ell^* 2^{j^*}z
     <  \left( 2^{j^*+1} \right) z/\eps,
\]
where the last inequality uses that~$\ell^* \leq 1/\eps$.
Hence,  $\optkz(P(t^-_{p^*}))=(L^* + 3\cdot 2^{j^*-1}z)/2$.

At time $t_{p^*}$ the point $p^*$ expires. Hence, at time~$t_{p^*}^+$
the subgroup~$G_{i^*}^{j^*,\ell^*}$ has only $z$ points, which can all be
designated as outliers. Hence, $C_{i^*}$ minus the outliers can be covered by a ball
of diameter $L^* + 2\cdot 2^{j^*-1}z$ and so $\optkz(P(t^+_{p^*}))=(L^* + 2\cdot 2^{j^*-1}z)/2$.

Since the algorithm must report a valid solution at time~$t_{p^*}^-$ and it
must report the same solution at time~$t_{p^*}^+$, the approximation ratio 
at time~$t_{p^*}^+$ is at least
\[
\frac{\optkz(P(t^-_{p^*}))}{\optkz(P(t^+_{p^*}))} 
    \geq \frac{L^*+3\cdot 2^{j^*-1}z}{L^* + 2\cdot 2^{j^*-1}z} 
      =  1 + \frac{2^{j^*-1}z}{L^*+2^{j^*}z}
      > 1 + \frac{2^{j^*-1}z}{2^{j^*+2}z/\eps}
      = 1 + \eps/8
\]
We conclude that any algorithm that gives a
$(1+\eps')$-approximation must store at least $k(z+1)(s-1)/(8\eps')$
objects in the worst case.

The spread of the point set used in the construction, including
the subgroup $G_{\mathrm{new}}$, is\footnote{In our definition of spread,
we are allowed to reuse points, since $\dmin=0$ is defined as the minimum
distance between distinct points in the stream. We can also avoid reusing points in our lower-bound scenario, without asymptotically influencing the spread of the points.}
\[
\sigma = kL + (k-1)2L + 2^{j^*}z < 3kL < 3k \cdot 2^{s+1}z/\eps. 
\]
If the spread is not too small, namely when $\sigma\geq (3kz/\eps)^2$,
we have $s\geq \frac{1}{2} \log \sigma - 1$.  We obtain the following theorem.
%------------------------------------------------------------------------------------------
\begin{theorem}\label{th:lower-bound-eps}
Let $0<\eps<1$ and let \alg be a $(1+\eps)$-approximation algorithm for the $k$-center problem with
$z$ outliers in $\Reals^1$, with $z\geq 1$, that works in the model described above.
For any spread~$\sigma \geq (3kz/\eps)^2$, there is a problem instance of spread~$\sigma$ 
that requires \alg to store $\Omega((kz\log\sigma)/\eps)$ points. 
\end{theorem}
%------------------------------------------------------------------------------------------

%------------------------------------------------------------------------------------------
\section{Concluding remarks}
%------------------------------------------------------------------------------------------
We presented the first algorithm for the $k$-center problem with outliers in the sliding-window model.
The algorithm works for points coming from a space of bounded doubling dimension.
Given any parameter $\eps>0$, the algorithm maintains a subset $S(t)$ of the set $P(t)$ 
of points that are currently in the window, such that $\optkz(S(t))\leq (1+\eps)\cdot\optkz(P(t))$.
The size of the set $S(t)$ is $O((kz/\eps^d)\log\sigma)$, where $\sigma$ is the spread of the points
in the stream.  We also showed that this amount of storage is optimal, except possibly for the dependency
on $\eps$, in the following sense: any deterministic $(1+\eps)$-approximation algorithm
for the 1-dimensional $k$-center problem with outliers in the sliding-window model must use
$\Omega((kz/\eps)\log\sigma)$ storage. This lower bound holds in a very general model: 
the main requirement is that the algorithm only updates its point set when a new point
arrives or at an explicitly stored expiration time of a previously seen point.

It would be interesting to investigate the dependency on the parameter $\eps$ in
more detail and close the gap between our upper and lower bounds.
A main question here is whether it is possible to develop a sketch whose storage is
only polynomially dependent on the doubling dimension~$d$.
Another interesting direction is to explore randomized algorithms.
Can we beat the above-mentioned lower bound using randomization?

\bibliography{references}

\end{document}